\documentclass[11pt,draftcls,onecolumn]{IEEEtran}
\usepackage{cite}

\usepackage{amsthm}
\usepackage{amssymb}
\usepackage{amsfonts}
\usepackage{color}

\newtheorem{proposition}{Proposition}
\newtheorem{lemma}{Lemma}
\newtheorem{theorem}{Theorem}
\newtheorem{remark}{Remark}

\newtheorem{definition}{Definition}

\ifCLASSINFOpdf
   \usepackage[pdftex]{graphicx}
   \graphicspath{{../pdf/}{../jpeg/}}
   \DeclareGraphicsExtensions{.pdf,.jpeg,.png}
\else
   \usepackage[dvips]{graphicx}
   \graphicspath{{../eps/}}
   \DeclareGraphicsExtensions{.eps}
\fi
\usepackage[cmex10]{amsmath}
\ifCLASSOPTIONcompsoc
  \usepackage[caption=false,font=normalsize,labelfont=sf,textfont=sf]{subfig}
\else
  \usepackage[caption=false,font=footnotesize]{subfig}
\fi
 \usepackage{dblfloatfix}
\begin{document}
%
\title{Wavelet-based Estimator for the Hurst Parameters of Fractional Brownian Sheet}
%
%
%

\author{Liang~Wu
        and~Yiming~Ding
\thanks{This work was supported in part by the National Basic Research Program of China (973 Program, grant No.2013CB910200, and grant No.2011CB707802).}
\thanks{Liang~Wu is with Wuhan Institute of Physics and Mathematics, Chinese Academy of Sciences, Wuhan 430071, China, and also with University of Chinese Academy of Sciences, Beijing 100049, China (e-mail: wuliangshine@gmail.com).}
\thanks{Yiming~Ding is with Key Laboratory of Magnetic Resonance in Biological Systems, Wuhan Institute of Physics and Mathematics, Chinese Academy of Sciences, Wuhan 430071, China (e-mail: ding@wipm.ac.cn).}
}

\maketitle

\begin{abstract}
It is proposed a class of statistical estimators $\hat H =(\hat H_1, \ldots, \hat H_d)$ for the Hurst parameters $H=(H_1, \ldots, H_d)$ of fractional Brownian field via multi-dimensional wavelet analysis and least squares, which are asymptotically normal.
These estimators can be used to detect self-similarity and long-range dependence in multi-dimensional signals, which is important in texture classification and improvement of diffusion tensor imaging (DTI) of nuclear magnetic resonance (NMR).
Some fractional Brownian sheets will be simulated and the simulated data are used to validate these estimators.
We find that when $H_i \geq 1/2$, the estimators are efficient, and when $H_i < 1/2$, there are some bias.
\end{abstract}


%
\IEEEpeerreviewmaketitle

\section{Introduction}

In numerous fields such as hydrology, biology, telecommunication and economics, the data available for analysis usually have scaling behavior (or long-range dependence and self-similarity) that need to be detected.
The key point for detecting scaling behavior is the estimation of the Hurst parameters $H$, which are used in characterizing self-similarity and long-range dependence.
The literature on the topic is vast (see e.g. \cite{FBMwave2,FBMwave1,FBMwave3,FBMwave4,FBMwave5,FBMwave6,FBMwave7,FBMwave8,Ding2014}).
Most of the published research concerns the scaling behavior in 1D data.
However, it is also important to detect the scaling behavior in multi-dimensional signals, which is important in texture classification and improvement of diffusion tensor imaging (DTI) of nuclear magnetic resonance (NMR).
Moreover, many multidimensional data from various scientific areas also have anisotropic nature in the sense that they have different geometric and probabilistic along different directions.
The study of long-range dependence and self-similarity in multidimensional case is usually based on the model of anisotropy fractional Gaussian fields such as fractional Brownian sheet (fBs) \cite{FBSdef1,Pesquet2002}, operator scaling Gaussian random field (OSGRF) \cite{Bierme2007} and extended fractional Brownian field (EFBF) \cite{Bonami2003}.

We focus on fBs, which plays an important role in modeling anisotropy random fields with self-similarity and long-range dependence (see e.g. \cite{Doukhan2003book,Pesquet2002,FBSstu3}) since it was first introduced by Kamont \cite{FBSdef1}.
Fractional Brownian sheet has arisen in the study of texture analysis for classification \cite{Atto2013,Roux2013}.
Besides research on the texture analysis, fractional Brownian sheets can also be used to drive stochastic partial differential equations \cite{book2}.
Moreover, Diffusion-tensor images (DTI) of nuclear magnetic resonance (NMR) in brain is now based on three-dimensional Brownian motion\cite{Moseley2002} $H_1=H_2=H_3$.
Compared with this model, fractional Brownian sheet may be a better model because it can describe the long-range dependence, which may exist in brain.

The main purpose of this paper is to estimate the Hurst parameters of fractional Brownian sheets via wavelet analysis.
Wavelet analysis has been an efficient tool for the estimation of Hurst parameter, but most of the published research on the topic focus on the 1D case.
Wavelet-based estimators for the parameter of single-parameter scaling processes (self-similarity and long-range dependence) have been proposed and studied by Abry et al. (see \cite{FBMwave2,FBMwave1,FBMwave3,FBMwave4,FBMwave5,FBMwave6,FBMwave7,FBMwave8}).
Such estimators are based on the wavelet expansion of the scaling process and a regression on the log-variance of wavelet coefficient.
It was usually assumed that wavelet coefficients are uncorrelated within and across octaves.
Bardet \cite{Bardet2002} and Morales \cite{Morales2002thesis} removed this assumption and proved the asymptotic normality of these estimators in the case of fractional Brownian motion (fBm), proposed a two-step estimator, which has the lowest variance in all of the least square estimators.
Compared with other estimators, such as the R/S method, the periodogram, the maximum-likelihood estimation, and so on, the wavelet-based estimator performs better in both the statistical and computational senses, and is superior in robustness (see \cite{FBMwave5,FBMwave6,FBMwave4} and the references therein).
It not only has small bias and low variance but also leads to a simple, low-cost, scalable algorithm.
Besides, the wavelet-based method can also eliminate some trends (linear trends, polynomial trend, or more) by the property of its vanishing moments \cite{FBMwave4}, which makes the estimator robust to some nonstationarities.

The paper is organized as follows.
In Section \ref{Sec:preliminaries}, some properties of Gaussian random variables and some limiting theorems used in this paper are introduced.
In Section \ref{Sec:estimator}, we obtain some properties of wavelet coefficient of fBs (see Proposition \ref{Prop:coef}).
Based on these properties, A class of estimators $\hat H$ for the Hurst parameter of fBs are proposed (see Section \ref{Sec:estiH}).
We also prove that these estimators $\hat H$ are asymptotically normal following the approach in \cite{Morales2002thesis} (see Theorem \ref{Thm:Hmain} and Section \ref{Sec:Hmainproof}).
Using the two-step procedure that proposed by Bardet \cite{Bardet2000}, we realize the generalized least squares (GLS) estimator, which has the lowest variance of $\hat H$ (see Section \ref{Sec:twostepesti}).
Finally, this two-step estimator $\hat H_{og}$ is applied to the synthetic fBs generated by the method of circulant embedding in the two-dimensional case and three-dimensional case (see Section \ref{Sec:simulation}).


\section{Preliminaries}\label{Sec:preliminaries}

\begin{definition}
Given a vector $H=(H_1, \ldots, H_d) \in (0,1)^d$, a fractional Brownian sheet (fBs) $B^H=\{ B^H(t), t\in \mathbb{R}^d_+ \}$ with Hurst parameter $H$ is a real-valued mean-zero Gaussian random fields with covariance function given by
\begin{eqnarray}\label{Eq:BHcov}
E(B^H (s)B^H (t)) = \frac{1}{{2^d }}\prod\limits_{i = 1}^d {(|s_i |^{2H_i }  + |t_i |^{2H_i }  - |s_i  - t_i |^{2H_i } )} , s,t \in \mathbb{R}^d_+.
\end{eqnarray}
\end{definition}
It follows that $B^H$ is an anisotropic Gaussian random field.

Note that if $d=1$, then $B^H$ is a fractional Brownian motion (fBm) with Hurst parameter $H_1\in (0,1)$, which as a moving-average Gaussian process is introduced by Mandelbrot and Van Ness \cite{FBMdef}. If $d>1$ and $H_1= \cdots = H_d =1/2$, then $B^H$ is the Brownian sheet.

Moreover, $B^H$ has a continuous version for all $H=(H_1, \ldots, H_d) \in (0,1)^d$ and is self-similar in the sense that for any diagonal matrix $A = {\rm diag}(a_1,a_2,\ldots,a_d)$, where $a_i>0$ for all $1\leq i \leq d$, the random field $\{ B^H(At),t \in \mathbb{R}^d \}$ has the same probability law as $\{ a^H B^H (t),t \in \mathbb{R}^d \}$, where $a^H=a^{H_1}_1 \cdots a^{H_d}_d$ (see e.g. \cite{FBSdef2}). $B^H$ also has stationary increments with respect to each variable (see e.g. \cite{FBSsim2}).

First we will use some properties of Gaussian random variables and some limiting theorems.

\begin{lemma}\label{Lm:covXY2}
$(X, Y)$ are two-dimensional mean-zero normal variable, i.e., $(X,Y) \sim N(0,\sigma _1 ^2 ;0,\sigma _2 ^2 ;\rho )$, then
\begin{align*}
{\mathop{\rm cov}} (X^2 ,Y^2 ) = 2{\mathop{\rm cov}} ^2 (X,Y).
\end{align*}
\end{lemma}

\begin{proof}
  Assume ${\rm cov}(X,Y)=v$.
  Let $Z= aX-Y$ s.t. ${\rm cov}(Z,Y)=0$,
  So we have $a=\sigma^2_2 / v$.
  $$\mathbb{E} Z^2 = a^2 \mathbb{E}X^2 + \mathbb{E}Y^2 -2a\mathbb{E}XY = a^2 \sigma^2_1 + \sigma^2_2 - 2av.$$
 \begin{align*}
  \mathbb{E}X^2Y^2 &=\frac{1}{a^2}\mathbb{E}Z^2Y^2 + \frac{2}{a}\mathbb{E}ZY^3 + \frac{1}{a^2}\mathbb{E}Y^4 \\
                   &=\frac{1}{a^2}\mathbb{E}Z^2Y^2 + \frac{2}{a}\mathbb{E}Z\mathbb{E}Y^3 + \frac{1}{a^2}\mathbb{E}Y^4.
  \end{align*}
 It's easy to check $\mathbb{E}Z\mathbb{E}Y^3=0$. 
 Put $a$ and $\mathbb{E} Z^2$ into the above formula, we have $\mathbb{E}X^2Y^2 = \sigma^2_1\sigma^2_2 + 2 v^2$. Then ${\mathop{\rm cov}} (X^2 ,Y^2 ) = 2{\mathop{\rm cov}} ^2 (X,Y)$.
\end{proof}

\begin{definition}
Let $f:\mathbb{R}^d \to \mathbb{R}$. $X=(X^{(1)}, \ldots, X^{(d)} )'$ is a Gaussian vector, and $f(X)$ has finite second moment, we define the Hermite rank of $f$ with respect to $X$ as
\begin{align*}
{\rm rank}(f):={\rm inf} \{\tau: \exists \; \text{ploynomial $P$ of degree $\tau$ with} \; \mathbb{E} [ (f(X)- \mathbb{E}f(X)) P(X^{(1)}, \ldots, X^{(d)}) ] \neq 0 \}.
\end{align*}
\end{definition}

The following two lemmas have been proved in \cite{Morales2002thesis} (see Lemma 4 and Lemma 5 of \cite{Morales2002thesis}).
\begin{lemma}\label{Lm:rank1}
If $X \sim N(0,\sigma^2)$, then for $q>-1/2$, the rank of $f_q(X)=|X|^q$ is 2.
\end{lemma}

\begin{lemma}\label{Lm:rank2}
Let $h_i(x)=f_q(x)=|x|^q, q>-1/2$ for every $1 \leq i \leq d$.
Let $t_1,\ldots,t_d$ be $d$ real numbers, and $X=(X^{(1)}, \ldots, X^{(d)} )'$  a zero-mean Gaussian vector. Then
\begin{align*}
h(X)= \sum\nolimits_{i=1}^d {t_i h_i( X^{(i)} )}
\end{align*}
is a function $g:\mathbb{R}^d \to \mathbb{R}$ of Hermite rank 2.
\end{lemma}

The following three limiting theorems have been proved in \cite{Basu2004book}, \cite{Arcones1994} and \cite{Serfling1980book} (see Theorem 9.0 of \cite{Basu2004book}, Theorem 4 of \cite{Arcones1994} and Theorem 3.3A of \cite{Serfling1980book}).

\begin{theorem}[\bf{Cram\'{e}r-Wold theorem}]\label{Thm:C-W}
Let $X_n=( X_n^{(1)} , \ldots, X_n^{(d)} )$ and $X=( X^{(1)}, \ldots, X^{(d)} )$ be random vectors. As $n \to +\infty$, $X_n \mathop\to\limits^d X$ if and only if:
\begin{align*}
\sum\nolimits_{i=1}^d {t_i X_n^{(i)}} \mathop\to\limits^d \sum\nolimits_{i=1}^d {t_i X^{(i)} }.
\end{align*}
for each $( t_1,\ldots,t_d ) \in \mathbb{R}^d$.
\end{theorem}

\begin{theorem}[\bf{Arcones, 1994}]\label{Thm:Arcones}
Let $\{ X_j \}_{j=1}^{\infty}$ be a stationary mean-zero Gaussian sequence of $\mathbb{R}^d$-valued vectors. Set $X_j=( X_j^{(1)}, \ldots, X_j^{(d)} )$. Let $f$ be a function on $\mathbb{R}^d$ with rank $\tau, 1 \leq \tau < \infty$. We define
\begin{align*}
r^{(i_1,i_2)}(k)=\mathbb{E}[X_m^{(i_1)} X_{m+k}^{(i_2)}],
\end{align*}
for $k\in \mathbb{Z}$, where $m$ is any large enough number such that $m,m+k \geq 1$. Suppose that
\begin{align}\label{Eq:ArconesR}
   \sum\nolimits_{k=-\infty}^{\infty} | r^{(i_1, i_2)}(k) |^{\tau} < \infty,
\end{align}
for each $1\leq i_1, i_2 \leq d$. Then
\begin{align}\label{Eq:ArconesN}
n^{-1/2}\sum\nolimits_{j=1}^{n}( f(X_j)-\mathbb{E}f(X_j) ) \mathop\to\limits^d N(0, \sigma^2),
\end{align}
where
\begin{align}\label{Eq:ArconesV}
\sigma^2 = {\rm Var} f(X_1) +2 \sum\nolimits_{k=1}^{\infty} {\rm Cov}(f(X_1), f(X_{1+k})).
\end{align}
\end{theorem}

\begin{theorem}\label{Thm:funclimit}
Let $X_n=( X_n^{(1)}, \ldots, X_n^{(d)} )$. Suppose that as $n \to \infty$,
\begin{align*}
\frac{X_n-\mu}{b_n} \mathop\to\limits^d N(\mathbf{0},\Sigma),
\end{align*}
where $\Sigma$ is covariance matrix and $ b_n \to \infty$. Let $g(x) = (g_1(x), \ldots, g_m(x)), x=(x_1, \ldots, x_d)$, be a vector-valued function for which each component function $g_i(x)$ is real-valued and has a nonzero differential $g_i(\mu; t), t=(t_1, \ldots, t_d),$ at $x=\mu$.Put
\begin{align*}
D=\left[ \frac{\partial g_i}{\partial x_j} |_{x=\mu} \right]_{m\times d}.
\end{align*}
Then
\begin{align*}
\frac{ g(X_n)-g(\mu) }{b_n} \mathop\to\limits^d N(\mathbf{0}, D \Sigma D').
\end{align*}
\end{theorem}

\section{Wavelet-based estimator}\label{Sec:estimator}

\subsection{Wavelet coefficients of fractional Brownian sheet}\label{Sec:wavecoeff}

Let $\psi(t) \in L^2(\mathbb{R})$ be an orthonormal wavelet with compact support and have $N\geq1$ vanishing moments.
A wavelet $\psi(t)$ is said to have a number $N \geq 1$ of vanishing moments if
\begin{eqnarray}
\int {t^k \psi (t)dt}  \equiv 0,\quad k = 0,1,2, \ldots ,N - 1, \; \text{and}\;   \int {t^N \psi (t)dt} \neq 0.
\end{eqnarray}
The dilations and translations of $\psi(t)$
\begin{eqnarray}
 \psi _{j,k} (t) = 2^{ - j/2} \psi (2^{ - j} t - k), \quad j,k \in \mathbb{Z}
\end{eqnarray}
construct the wavelet orthonormal basis for $L^2(\mathbb{R})$. The factors $2^j$ and $j$ are called the scale and octave respectively. The factor $k$ is called translation. The wavelet orthonormal basis for $L^2(\mathbb{R}^d)$ can be obtained by simply taking the tensor product functions generated by $d$ one-dimensional bases: \cite{book3,book4}
\begin{eqnarray}
\psi _{j_1 , \ldots, j_d ,k_1 , \ldots, k_d } (t_1 ,t_2 , \ldots ,t_d ) = \psi _{j_1 ,k_1 } (t_1 )\psi _{j_2 ,k_2 } (t_2 ) \cdots \psi _{j_d ,k_d } (t_d ).
\end{eqnarray}
For convenience, let $J = (j_1 , \ldots ,j_d )$ $\in \mathbb{Z}^d$, $K = (k_1 , \ldots ,k_d )$ $\in \mathbb{Z}^d$, $t =(t_1 , \ldots ,t_d )\in \mathbb{R}^d$. Then the wavelet coefficients of fractional Brownian sheet $B^H$ in $\mathbb{R}^d$ are:
\begin{eqnarray}\label{Eq:wavecoeff}
d_{B^H}(J,K) := \int_\mathbb{R} { \cdots \int_\mathbb{R} {\psi _{j_1 ,k_1 } (t_1 )\psi _{j_2 ,k_2 } (t_2 ) \cdots \psi _{j_d ,k_d } (t_d )B^H (t)dt_1 }  \cdots dt_d }.
\end{eqnarray}
Before introducing our estimator for $B^H$, we present some properties of the  wavelet coefficient of $B^H$.

\begin{proposition}\label{Prop:coef}
 Let $\psi(t) \in L^2(\mathbb{R})$ be an orthonormal with compact support and have $N\geq1$ vanishing moments. The wavelet coefficients of fractional Brownian sheet $B^H$ defined by Eq.(\ref{Eq:wavecoeff}) have the following properties:

 (\romannumeral1) $\mathbb{E}d_{B^H}(J,K) = 0$, and $d_{B^H}(J,K)$ is Gaussian, for any $J,K \in \mathbb{Z}^d$.

 (\romannumeral2) For fixed $J= (j_1 ,j_2 , \ldots ,j_d ) \in \mathbb{Z}^d$, and  $ \forall   K_1, \ldots, K_n \in \mathbb{Z}^d$, we have
  \begin{eqnarray}\label{Eq:coefscaling}
( d_{B^H } (J,K_1), \ldots, d_{B^H } (J,K_n) ) \mathop = \limits^d 2^{\sum\nolimits_{i = 1}^d {j_i (H_i  + 1/2)} } ( d_{B^H } (\mathbf{0}, K_1), \ldots, d_{B^H } (\mathbf{0}, K_n) ),
  \end{eqnarray}
where $\mathbf{0}$ denotes $(0,0,\cdots,0)\in \mathbb{Z}^d$.

 (\romannumeral3) For fixed $J \in \mathbb{Z}^d$, and  $ \forall  K_1, \ldots, K_n, L \in \mathbb{Z}^d$, we have
   \begin{eqnarray}\label{Eq:coefstationary}
( d_{B^H } (J,K_1+L), \ldots, d_{B^H } (J,K_n+L) ) \mathop = \limits^d ( d_{B^H } (J, K_1), \ldots, d_{B^H } (J, K_n) ),
  \end{eqnarray}

 (\romannumeral4) For fixed $J=(j_1, \ldots, j_d ) \in \mathbb{Z}^d$, $J'=(j'_1, \ldots, j'_d )\in \mathbb{Z}^d$, $K=(k_1, \ldots, k_d ) \in \mathbb{Z}^d$, and $K'=(k'_1, \ldots, k'_d ) \in \mathbb{Z}^d$, when $\mathop {\min }\limits_{1 \le i \le d} |2^{j_i } k_i  - 2^{j'_i } k'_i | \to  + \infty$
 \begin{eqnarray}\label{Eq:coefcorrelation}
\mathbb{E}d_{B^H } (J,K)d_{B^H } (J',K') \approx \prod\limits_{i = 1}^d {|2^{j_i } k_i  - 2^{j'_i } k'_i |^{2H_i  - 2N} }.
 \end{eqnarray}
In the above, $ \mathop=\limits^d$ means equality in distribution.
\end{proposition}

\begin{proof}
(\romannumeral1) $$\mathbb{E}d_{B^H } (J,K) = \int_{\mathbb{R}^d} {\mathbb{E}B^H (t)\psi _{j_1 ,k_1 } (t_1 )\psi _{j_2 ,k_2 } (t_2 ) \cdots \psi _{j_d ,k_d } (t_d )dt},$$
 $\mathbb{E}B^H (t)$ is a constant, and $\int {\psi _{j,k} (t)dt}  = 0$, so $\mathbb{E}d_{B^H}(J,K) = 0$. The Gaussianity of  $d_{B^H}(J,K)$ follows from the fact that the wavelet coefficient of Gaussian process is Gaussian (see \cite{FBMwave9}).

(\romannumeral2) This property can be obtained by the self-similarity of $B^H$. For $\forall K \in \mathbb{Z}^d$
\begin{align*}
 d_{B^H } (J,K) &= 2^{ - \frac{1}{2}\sum\nolimits_{i = 1}^d {j_i } } \int_{\mathbb{R}^d} {\psi (2^{ - j_1 } t_1  - k_1 ) \cdots \psi (2^{ - j_d } t_d  - k_d )B^H (t)dt}  \\
  &= 2^{\frac{1}{2}\sum\nolimits_{i = 1}^d {j_i } } \int_{\mathbb{R}^d} {\psi (t'_1  - k_1 ) \cdots \psi (t'_d  - k_d )B^H (2^{j_1 } t'_1 , \cdots ,2^{j_d } t'_d )dt'}  \\
 &\mathop  = \limits^d 2^{\sum\nolimits_{i = 1}^d {j_i (H_i  + 1/2)} } d_{B^H } (\mathbf{0},K),
\end{align*}
where $t'=(t'_1,\ldots, t'_d)$.

Similarly, one can check that, for $\theta_i \in \mathbb{R}, 1 \leq i \leq n$,
\begin{align*}
\sum\nolimits_{i = 1}^n {\theta _i d_{B^H } (J,K_i )}  \mathop  = \limits^d  2^{\sum\nolimits_{i = 1}^d {j_i (H_i  + 1/2)} } \sum\nolimits_{i = 1}^n {\theta _i d_{B^H } (\mathbf{0},K_i )}.
\end{align*}

(\romannumeral3) Let $d(J, k_1 , \ldots ,k_d ) = d(J,K)$, $L=(l_1,\ldots,l_d)$. Without loss of generality, we can take $J=\mathbf{0}$. First we check the formula below,
\begin{eqnarray}\label{Eq:d0kh}
  d_{B^H }(\mathbf{0},k_1  + l_1 , \ldots ,k_d  + l_d )\mathop  = \limits^d d_{B^H }(\mathbf{0}, k_1 ,k_2  + l_2 , \ldots ,k_d  + l_d ).
\end{eqnarray}
In fact, 
\begin{align*}
  &d_{B^H }(\mathbf{0},k_1  + l_1 , \ldots ,k_d  + l_d ) \\
  =& \int_{\mathbb{R}^d} {\psi (t_1  - k_1  - l_1 ) \cdots \psi (t_d  - k_d  - l_d )B^H (t_1 , \ldots ,t_d )dt_1  \cdots dt_d }  \\
  =&\int_{\mathbb{R}^{d - 1} } \int_\mathbb{R} {\psi (t'_1  - k_1 )B^H (t'_1  + l_1 ,t_2  \ldots ,t_d )dt'_1 } \cdot \prod\limits_{i = 2}^d {\psi (t_i  - k_i  - l_i )} dt_2  \cdots dt_d .
\end{align*}
Since $\int_\mathbb{R} {\psi (t'_1  - k_1 )dt'_1 }  = 0$ and $B^H$ has stationary increments with respect to each variable, we have
\begin{align*}
 \int_\mathbb{R} {\psi (t'_1  - k_1 )B^H (t'_1  + l_1 ,t_2  \ldots ,t_d )dt'_1 } \mathop  = \limits^d\int_\mathbb{R} {\psi (t'_1  - k_1 )B^H (t'_1 ,t_2  \ldots ,t_d )dt'_1 }.
\end{align*}

The Eq.(\ref{Eq:d0kh}) is proved. 

Applying Eq.(\ref{Eq:d0kh}) repeatedly, 
\begin{align*}
 d_{B^H } (\mathbf{0},K+L) &\mathop  = \limits^d d_{B^H } (\mathbf{0},k_1 ,k_2  + l_2 , \ldots ,k_d  + l_d )\mathop  = \limits^d d_{B^H } (\mathbf{0},k_1 ,k_2 , k_3+l_3, \ldots ,k_d  + l_d ) \\
                             &\mathop  = \limits^d  \ldots \mathop  = \limits^d d_{B^H } (\mathbf{0},k_1 , \ldots ,k_d ) = d_{B^H } (\mathbf{0},K).
\end{align*}

Similarly, one can see that, for $\theta_i \in \mathbb{R}, 1 \leq i \leq n$,
\begin{align*}
\sum\nolimits_{i = 1}^n {\theta _i d_{B^H } (\mathbf{0},K_i+L )}  \mathop  = \limits^d  \sum\nolimits_{i = 1}^n {\theta _i d_{B^H } (\mathbf{0},K_i )}.
\end{align*}

(\romannumeral4) The covariance function of $B^H$ is given by Eq.(\ref{Eq:BHcov}). And
\begin{align*}
 &\mathbb{E}d_{B^H } (J,K)d_{B^H } (J',K') \\
  =& \int_{\mathbb{R}^d} {\int_{\mathbb{R}^d} {\mathbb{E}B^H (t)B^H (s)\prod\limits_{i = 1}^d {\psi _{j_i ,k_i } (t_i )\psi _{j'_i ,k'_i } (s_i )} dsdt} }.
\end{align*}
So we have:
\begin{align}
 &\mathbb{E}d_{B^H } (J,K)d_{B^H } (J',K') \notag \\
  =& \prod\limits_{i = 1}^d {\int_\mathbb{R} {\int_\mathbb{R} {\psi _{j_i ,k_i } (t_i )\psi _{j'_i ,k'_i } (s_i )\frac{1}{2}(|s_i |^{2H_i }  + |t_i |^{2H_i }  - |s_i  - t_i |^{2H_i } )ds_i dt_i } } } \notag \\
  =& \prod\limits_{i = 1}^d { \mathbb{E}d_{B^{H_i} } (j_i,k_i)d_{B^{H_i} } (j'_i,k'_i)  }, \label{Eq:P1cov}
\end{align}
where $B^{H_i}$ is fBm with Hurst exponent $H_i$.

It has been shown that the correlations between wavelet coefficients of fBm satisfy this asymptotic equation \cite{FBMwave11a,FBMwave11b,FBMwave6}
\begin{eqnarray*}
\mathbb{E}d_{B^H } (j,k)d_{B^H } (j',k') \approx |2^j k - 2^{j'} k'|^{2H - 2N} ,|2^j k - 2^{j'} k'| \to \infty,
\end{eqnarray*}
where $j,j',k,k' \in \mathbb{Z}$, and $B^H$ is fBm with Hurst exponent $H$, and the wavelet used has compact support.

The Eq.(\ref{Eq:coefcorrelation}) is proved.
\end{proof}

\begin{remark}\label{Rmk:N}
In view of Eq.(\ref{Eq:coefcorrelation}), to avoid long-range dependence for $d_{B^H}(J,K)$, i.e., to ensure that $\sum\nolimits_{J,K \in \mathbb{Z}^d } {\mathbb{E}|d_{B^H } (J,K)d_{B^H }(\mathbf{0},\mathbf{0})|}  < \infty$, one can choose
\begin{eqnarray}\label{Eq:N}
 N > \mathop {\max }\limits_{1 \le i \le d} H_i  + 1/2, i.e., \mathop {\max }\limits_{1 \le i \le d} (2H_i  - 2N) <  - 1.
\end{eqnarray}
Under this condition, the correlation of $d_{B^H } (J,K)$ tends rapidly to $0$ at large lags.
\end{remark}

\subsection{Estimation of the Hurst parameter and its asymptotic behavior}\label{Sec:estiH}

In this subsection, we estimate the Hurst parameter of fractional Brownian sheet in the $d$-dimensional finite interval.
The following lemma can be obtained by direct calculation.

\begin{lemma}\label{Lm:Tnj}
$\{ B^H(t), t \in \prod^{d}_{i=1}[0, T_i], T=(T_1, \ldots, T_d) \in \mathbb{R}_+^d  \}$ is a fractional Brownian sheet in the $d$-dimensional finite interval.
$\psi(t)$ is an orthogonal wavelet with compact support $[-M, M]$, $M>0$ and has $N\geq 1$ vanishing moments.
We define:
\begin{align}\label{Eq:dT}
d_{B^H,T}(J,K) := \int_0^{T_d} { \cdots \int_0^{T_1} {\psi _{j_1 ,k_1 } (t_1 )\psi _{j_2 ,k_2 } (t_2 ) \cdots \psi _{j_d ,k_d } (t_d )B^H (t) dt_1 }  \cdots dt_d }.
\end{align}
For ease of writing, let $ M \leq 1$. So when $ 1 \leq k_i \leq 2^{-j_i}T_i - M$, $i=1, \ldots, d$,
\begin{align}\label{Eq:dTd}
d_{B^H,T}(J,K) = d_{B^H}(J,K),
\end{align}
where $d_{B^H}(J,K)$ is defined by Eq.(\ref{Eq:wavecoeff}).
\end{lemma}

\begin{definition}\label{def:acoff}
Under the conditions of Lemma \ref{Lm:Tnj},
we define that the available wavelet coefficients for the finite interval fractional Brownian sheet are
the wavelet coefficients $d_{B^H,T}(J,K)$ that satisfy Eq.(\ref{Eq:dTd}).
Let $n_J$ be the number of available wavelet coefficients at octave $J\in \mathbb{Z}^d$.
\begin{align*}
n_J = \prod\nolimits_{i=1}^d n_i, \;\; n_i := \lfloor 2^{-j_i}T_i - M \rfloor.
\end{align*}
If  there exists $i_0 \in \{1, \ldots, d \}$ s.t. $T_{i_0} \to +\infty $, then $n_J \to +\infty$.
\end{definition}

Under the conditions of Lemma \ref{Lm:Tnj},
according to Eq.(\ref{Eq:coefscaling}), Eq.(\ref{Eq:coefstationary}) and Lemma \ref{Lm:Tnj},
the wavelet coefficients of $B^H(t)$ satisfy the equation,
\begin{eqnarray}\label{Eq:2law}
\mathbb{E}d^2_{B^H,T} (J,K) = \mathbb{E}d^2_{B^H } (J,K) = C_{B^H} 2^{j_1 (2H_1  + 1) + j_2 (2 H_2  + 1) +  \cdots  + j_d (2H_d  + 1)},
\end{eqnarray}
where the Hurst parameter $H= (H_1,\ldots,H_d) \in (0,1)^d$, $C_{B^H}= \mathbb{E} d^2_{B^H}(\mathbf{0},\mathbf{1})$ only depends on $B^H$, $J=(j_1,\ldots,j_d), K=(k_1,\ldots,k_d) \in \mathbb{Z}^d$.

Take the 2-logarithm of both sides:
\begin{eqnarray}\label{Eq:loglaw}
\log _2 \mathbb{E} d^2_{B^H,T} (J,K) = j_1 (2 H_1  + 1) +  \cdots  + j_d (2 H_d  + 1) + \log _2 C_{B^H}.
\end{eqnarray}

So the estimation of $H$ can be realized by a linear regression in the left part versus $J$ diagram. The $\mathbb{E} d^2_{B^H,T}(J,K)$ can be estimated by
\begin{align*}
S(J):=1/n_J \sum\nolimits_{k_d=1}^{n_d} { \cdots \sum\nolimits_{k_1=1}^{n_1} { d^2_{B^H,T}(J,K)} }.
\end{align*}
According to Lemma \ref{Lm:Tnj},
\begin{align}\label{Eq:Sj}
S(J) =1/n_J \sum\nolimits_{k_d=1}^{n_d} { \cdots \sum\nolimits_{k_1=1}^{n_1} { d^2_{B^H}(J,K)} }.
\end{align}

In more detail, the $H$ is estimated by the regression
\begin{align*}
\left( {\begin{array}{*{20}c}
   {\log _2 S (J_1 )}  \\
    \vdots   \\
   {\log _2 S (J_m )}  \\
\end{array}} \right) = \left( {\begin{array}{*{20}c}
   {\begin{array}{*{20}c}
   {J_1 }  \\
    \vdots   \\
   {J_m }  \\
\end{array}} & {\begin{array}{*{20}c}
   1  \\
    \vdots   \\
   1  \\
\end{array}}  \\
\end{array}} \right)\left( {\begin{array}{*{20}c}
   {2H_1  + 1}  \\
    \vdots   \\
   {2H_d  + 1}  \\
   {\log _2 C_{B^H} }  \\
\end{array}} \right) + \left( {\begin{array}{*{20}c}
   {\varepsilon(J_1 )}  \\
    \vdots   \\
   {\varepsilon  (J_m )}  \\
\end{array}} \right),
\end{align*}
where $J_l=(j_{l,1},\ldots,j_{l,d}) \in \mathbb{Z}^d_+, l \in \{1,\ldots,m \}$, and $J_1 < \cdots < J_m$, ( $J_l<J_k$ in the sense that $j_{l,i} \leq j_{k,i}$, for $i \in \{1,\ldots,d \}$  and the equality can't hold for all $i$ ).

Let
\begin{align*}
L := \left( {\begin{array}{*{20}c}
   {\log _2 S (J_1 )}  \\
    \vdots   \\
   {\log _2 S (J_m )}  \\
\end{array}} \right), A := \left( {\begin{array}{*{20}c}
   {\begin{array}{*{20}c}
   {J_1 }  \\
    \vdots   \\
   {J_m }  \\
\end{array}} & {\begin{array}{*{20}c}
   1  \\
    \vdots   \\
   1  \\
\end{array}}  \\
\end{array}} \right),\alpha  := \left( {\begin{array}{*{20}c}
   {2H_1  + 1}  \\
    \vdots   \\
   {2H_d  + 1}  \\
   {\log _2 C_{B^H} }  \\
\end{array}} \right),\varepsilon {\rm{ := }}\left( {\begin{array}{*{20}c}
   {\varepsilon (J_1 )}  \\
    \vdots   \\
   {\varepsilon  (J_m )}  \\
\end{array}} \right),
\end{align*}
suppose that $m \geq d+1$, ${\rm rank} (A) = d+1$.

Then least squares estimation gives the estimator:
\begin{align}
 & \hat \alpha = (A' \Sigma^{-1} A)^{ - 1} A \Sigma^{-1} L, \notag \\
 & \hat H  = \hat \alpha/2  - 1/2,\label{Eq:Hesti}
\end{align}
where $\hat H  = (\hat H_1, \cdots, \hat H_d, \widehat{\log_2 C_{B^H}} /2 -1/2 )'$. $\hat H_i$ is the wavelet-based estimator of $H_i$, $i \in \{1,\ldots,d \}$, $\Sigma$ is a $m\times m$ full-rank matrix. $\hat H$ is selected to minimize
\begin{align*}
 \left\| {L  - A\alpha  } \right\|_{\Sigma} ^2  = (L - A\alpha)' \Sigma^{-1}  (L  - A\alpha ).
\end{align*}

If $\Sigma = I_{m\times m}$,  $\hat H$ is the ordinary least squares (OLS) estimator.
If $\Sigma = {\mathop{\rm Var}} L = {\mathop{\rm Var}} \varepsilon$,  $\hat H$ is the generalized least squares (GLS) estimator.

The following central limit theorem indicates the asymptotic normality of the estimators $\hat H$.
\begin{theorem}\label{Thm:Hmain}
Under the conditions of Lemma \ref{Lm:Tnj}, $\hat H$ is defined by Eq.(\ref{Eq:Hesti}),
$N>\mathop{\max }\limits_{1 \le i \le d} H_i+1/4$.
For any $c \in \{1, \ldots, d \}$, $T_c \to +\infty $, so $n_{l,c}\to +\infty$ and $n_{J_l}\to +\infty$, $ l \in \{1, \ldots, m \}$ ($n_{J_l}=  \prod\nolimits_{i=1}^d n_{l,i}$ is the same as that in Definition \ref{def:acoff}),
we have
\begin{align}
n_{m,c} ^{1/2} (\hat H - U) \mathop\to\limits^d N(0,\Sigma _{\hat H}),
\end{align}
where $U=(H',\log_2 C_{B^H}/2-1/2)'$ and
\begin{align}\label{Eq:sigmaH}
\Sigma _{\hat H} =\frac{1}{4}(A' \Sigma^{-1} A)^{ - 1} A' \Sigma^{-1} \Sigma _{L} \Sigma^{-1} A (A' \Sigma^{-1} A)^{ - 1},
\end{align}
with $\Sigma _{L }  = D' \Sigma _{S } D$, $\Sigma _{S}  = (\sigma ^2 _{l,h} )_{m \times m}$, $\sigma ^2 _{l,h}$ is defined by (\ref{Eq:covSqlh}),
 and
\begin{align*}
D = \log _2 e \cdot {\rm diag}(\left( { \mathbb{E} d^2_{B^H } (J_1 ,K)} \right)^{ - 1} , \ldots ,\left( { \mathbb{E} d^2_{B^H } (J_m ,K) } \right)^{ - 1} ).
\end{align*}
\end{theorem}

\begin{remark}\label{Rmk:sigma}
$\Sigma _{L } / n_{m,c} $ is the asymptotic covariance of $L$.
For GLS estimator, $\Sigma$ is the covariance matrix of $L$. If $T_c \to +\infty $, $\Sigma \to \Sigma _{L}/ n_{m,c}$, then
\begin{align}
\Sigma _{\hat H} = \frac{1}{4} (A' \Sigma_{L}^{-1} A)^{ - 1}.
\end{align}
\end{remark}
\begin{remark}
According to the equation between $T_c$ and $n_{m,c}$ (see Definition \ref{def:acoff}), Theorem \ref{Thm:Hmain} indicates that the convergence rate of $\hat H$ is $O(1/\sqrt{T_c})$ as $T_c \to +\infty $ for any $c \in \{1, \ldots, d \}$.
Similarly, we can also obtain that the convergence rate of $\hat H$ is $O(1/\sqrt{\prod\nolimits_{i=1}^d {T_i}})$ as $T_i \to +\infty $ for all $i \in \{1, \ldots, d \}$.
That is to say, the estimation accuracy of $\hat H$ can be improved by increasing the square root of sample volume $\sqrt{\prod\nolimits_{i=1}^d {T_i}}$.
\end{remark}

%

\subsection{Proof of Theorem \ref{Thm:Hmain}} \label{Sec:Hmainproof}

The following limit theorem of $S(J)$ can be obtained by using Theorem \ref{Thm:Arcones}.

\begin{lemma}\label{Lm:SqJ}
Under the conditions of Lemma \ref{Lm:Tnj}, $S(J)$ is defined by Eq.(\ref{Eq:Sj}),
$N>\mathop{\max }\limits_{1 \le i \le d} H_i+1/4$.
For any $c \in \{1, \ldots, d \}$, $T_c \to +\infty $, so $n_c\to +\infty$ and $n_J\to +\infty$, we have
\begin{align}
n_c ^{1/2} [S(J) - u (J)] \mathop\to\limits^d N(0,\sigma^2 (J)),
\end{align}
where $n_J$ is the same as that in Definition \ref{def:acoff}, $u (J) = \mathbb{E} d^2_{B^H} (J,K)$,
\begin{align}
 \sigma  ^2 (J) &= \frac{1}{(n_J /n_c )^2}{\mathop{\rm Var}} (\sum\limits_{k_c  = 1,k_i  \in \{ 1, \ldots n_i \} ,i \ne c} d^2_{B^H } (J,K) ) \notag\\
  &+ \frac{2}{(n_J /n_c )^2}\sum\limits_{k = 1}^\infty  {{\mathop{\rm Cov}} (\sum\limits_{k_c  = 1,k_i  \in \{ 1, \ldots n_i \} ,i \ne c} d^2_{B^H } (J,K) ,\sum\limits_{k_c  = 1 + k,k_i  \in \{ 1, \ldots n_i \} ,i \ne c} d^2_{B^H } (J,K))}.
\end{align}
\end{lemma}

\begin{proof}
We construct a sequence of random vectors $X_v$ from the wavelet coefficients $d_{B^H } (J,K)$,
\begin{align*}
X_v  = (X_v (1), \ldots ,X_v (n_J /n_c ))': = [d_{B^H } (J,K)|_{k_c  = v, k_i  \in \{ 1, \ldots n_i \}, i \ne c } ]_{(n_J /n_c ) \times 1}, \; 1 \leq v \leq n_c.
\end{align*}
When $n_c\to +\infty$, using Proposition \ref{Prop:coef}, $\{ X_v \}_{v=1}^{\infty}$ is a stationary mean-zero Gaussian sequence.
Define
\begin{align*}
h(X_v ) = \frac{1}{{n_J /n_c }}\sum\nolimits_{i = 1}^{n_J /n_c } {|X_v (i)|^2 }.
\end{align*}
According to Lemma \ref{Lm:rank2}, the Hermite rank of $h(X_v )$ is 2.
So by Proposition \ref{Prop:coef}, when $N>\mathop{\max }\limits_{1 \le i \le d} H_i+1/4$, for each $1\leq i_1, i_2 \leq n_J /n_c $,
\begin{align*}
   \sum\nolimits_{k=-\infty}^{\infty} | r^{(i_1, i_2)}(k) |^2 < \infty,
\end{align*}
where $r^{(i_1,i_2)}(k)=\mathbb{E}[X_m(i_1) X_{m+k}(i_2)]$.

Hence, by Theorem \ref{Thm:Arcones}, the lemma is true.
\end{proof}

Now consider the asymptotic behavior of the full vector $S = (S (J_1),\ldots,S(J_m))'$. 
In order to apply Theorem \ref{Thm:C-W} (Cram\'{e}r-Wold theorem), we need to study the asymptotic behavior of $a S$, $a=(a_1,\ldots,a_m)$ is a fixed but arbitrary element of $\mathbb{R}^m$.

\begin{lemma}\label{Lm:aSq}
Under the conditions of Lemma \ref{Lm:Tnj}, let $a=(a_1,\ldots,a_m)$ be a fixed but arbitrary element of $\mathbb{R}^m$, $S = (S (J_1),\ldots,S(J_m))'$.
$N>\mathop{\max }\limits_{1 \le i \le d} H_i+1/4$.
For any $c \in \{1, \ldots, d \}$, $T_c \to +\infty $, $n_{J_l}=  \prod\nolimits_{i=1}^d n_{l,i}$, $u (J_l) = \mathbb{E} d^2_{B^H} (J_l,K)$, $ l \in \{1, \ldots, m \}$, we have
\begin{align}
n_{m,c } ^{1/2} [aS - \sum\nolimits_{l = 1}^m {a_l u (J_l )} ]\mathop  \to \limits^d N(0,\sigma _{aS } ^2 ),
\end{align}

\begin{align}
 & \sigma _{aS } ^2  = {\mathop{\rm Var}} \left( {\sum\limits_{l = 1}^m {a_l \frac{{n_{l,c} }}{{u_l n_{J_l } }}\sum\limits_{\scriptstyle k_c  \in \{ 1, \ldots ,u_l \} , \hfill \atop
  \scriptstyle k_i  \in \{ 1, \ldots ,n_{l,i} \} ,i \ne c \hfill} {d^2_{B^H } (J_l ,K ) } } } \right) +     \notag\\
 &2\sum\limits_{k = 1}^\infty  {{\mathop{\rm Cov}} \left( {\sum\limits_{l = 1}^m {a_l \frac{{n_{l,c} }}{{u_l n_{J_l } }}\sum\limits_{\scriptstyle k_c  \in \{ 1, \ldots ,u_l \} , \hfill \atop
  \scriptstyle k_i  \in \{ 1, \ldots ,n_{l,i} \} ,i \ne c \hfill} {d^2_{B^H } (J_l ,K ) } } ,\sum\limits_{l = 1}^m {a_l \frac{{n_{l,c} }}{{u_l n_{J_l } }}\sum\limits_{\scriptstyle k_c  \in \{ u_l k+1, \ldots ,u_l (k+1)\} , \hfill \atop
  \scriptstyle k_i  \in \{ 1, \ldots ,n_{l,i} \} ,i \ne c \hfill} {d^2_{B^H } (J_l ,K ) } } } \right)},
\end{align}
and $K=(k_1, \ldots,  k_d)$, $u_l  = 2^{j_{m,c}-j_{l,c}} $, $l \in \{1,\ldots,m \}$.
\end{lemma}

\begin{proof}
 We construct a sequence of random vectors $X_{l,s}$ from the wavelet coefficients $d_{B^H } (J_l,K)$, $K=(k_1, \ldots,  k_d)$.
\begin{align*}
X_{l,s}  = (X_{l,s} (1), \ldots ,X_{l,s} (n_{J_l} /n_{l,c} ) ): = [d_{B^H } (J_l,K)|_{k_c  = s, k_i  \in \{ 1, \ldots n_{l,i} \}, i \ne c } ]_{1\times (n_{J_l} /n_{l,c} )},\;\; 1\leq s\leq n_{l,c}.
\end{align*}
Because $J_1 < \cdots < J_m$, according to Definition \ref{def:acoff}, $n_{1,c} \geq \cdots  \geq n_{m,c}$.
One can check $u_l  \approx \frac{n_{l,c} }{n_{m,c} }$.
Without loss of generality, we suppose that $ n_{l,c} = u_l n_{m,c}$, $l \in \{1, \ldots, m \}$.

Let
\begin{align*}
X_v  = ( X_{1, u_1 (v-1)+1},\ldots, X_{1, u_1 v},\ldots, X_{m, u_m (v-1)+1},\ldots, X_{m, u_m v} )', 1\leq v \leq n_{m,c}.
\end{align*}
When $T_c \to +\infty$, $n_{m,c} \to +\infty$, using Proposition \ref{Prop:coef}, $\{ X_v \}_{v=1}^{\infty}$ is a stationary mean-zero Gaussian sequence.
Define
\begin{align*}
h(X_v ) = \sum\limits_{l = 1}^m {a_l \frac{1}{{u_l n_{J_l } /n_{l,c} }}\sum\limits_{j = 1}^{u_l } {\sum\limits_{i = 1}^{n_{J_l } /n_{l,c} } {|X_{l,u_l (v - 1) + j} (i)|^2 } } }.
\end{align*}
By Lemma \ref{Lm:rank2}, the Hermite rank of $h(X_v )$ is 2.

We rewrite $X_v$ as
\begin{align*}
X_v  = ( X_v(1),\ldots,X_v( \sum\nolimits_{l = 1}^m { u_l n_{J_l } /n_{l,c}} ) )'.
\end{align*}
So by Proposition \ref{Prop:coef}, when $N>\mathop{\max }\limits_{1 \le i \le d} H_i+1/4$, it's easy to check, for each $1\leq i_1, i_2 \leq \sum\nolimits_{l = 1}^m { u_l n_{J_l } /n_{l,c} } $,
\begin{align*}
   \sum\nolimits_{k=-\infty}^{\infty} | r^{(i_1, i_2)}(k) |^2 < \infty,
\end{align*}
where $r^{(i_1,i_2)}(k)=\mathbb{E}[X_m(i_1) X_{m+k}(i_2)]$.

Hence, by Theorem \ref{Thm:Arcones}, Lemma \ref{Lm:aSq} is proved.
\end{proof}

So by Theorem \ref{Thm:C-W} (Cram\'{e}r-Wold theorem), the asymptotic behavior of the vector $S$ can be obtained.

\begin{lemma}\label{Lm:Sq}
Under the conditions of Lemma \ref{Lm:Tnj}, $S = (S (J_1),\ldots,S(J_m))'$.
$N>\mathop{\max }\limits_{1 \le i \le d} H_i+1/4$.
For any $c \in \{1, \ldots, d \}$, $T_c \to +\infty $, we have
\begin{align}
n_{m,c } ^{1/2} [S  -  u  ]\mathop \to\limits^d  N(0,\Sigma _{S }),
\end{align}
where
\begin{align}
u  = (\mathbb{E}d^2_{B^H } (J_1 ,K ) , \ldots ,\mathbb{E}d^2_{B^H } ( J_m ,K ) ),
\end{align}
and
\begin{align}\label{Eq:sigmaS1}
\Sigma _{S }  = (\sigma ^2 _{l,h} )_{m \times m},
\end{align}
with
\begin{align}
 & \sigma _{l,h} ^2  = \frac{1}{{u_l n_{J_l } /n_{l,c} }}\frac{1}{{u_h n_{J_h } /n_{h,c} }}  [ {\mathop{\rm Cov}} \left( {\sum\limits_{\scriptstyle k_c  \in \{ 1, \ldots ,u_l \} , \hfill \atop
  \scriptstyle k_i  \in \{ 1, \ldots ,n_{l,i} \} ,i \ne c \hfill} d^2_{B^H } (J_l ,K) ,\sum\limits_{\scriptstyle k_c  \in \{ 1, \ldots ,u_h \} , \hfill \atop
  \scriptstyle k_i  \in \{ 1, \ldots ,n_{h,i} \} ,i \ne c \hfill} d^2_{B^H } (J_h ,K) } \right) \notag \\
 &+ \sum\limits_{k = 1}^\infty  {{\mathop{\rm Cov}} \left( {\sum\limits_{\scriptstyle k_c  \in \{ 1, \ldots ,u_l \} , \hfill \atop
  \scriptstyle k_i  \in \{ 1, \ldots ,n_{l,i} \} ,i \ne c \hfill} d^2_{B^H } (J_l ,K) ,\sum\limits_{\scriptstyle k_c  \in \{ u_h k + 1, \ldots ,u_h (k + 1)\} , \hfill \atop
  \scriptstyle k_i  \in \{ 1, \ldots ,n_{h,i} \} ,i \ne c \hfill} d^2_{B^H } (J_h ,K)} \right)} \notag \\
 &+ \sum\limits_{k = 1}^\infty  {{\mathop{\rm Cov}} \left( {\sum\limits_{\scriptstyle k_c  \in \{ u_l k + 1, \ldots ,u_l (k + 1)\} , \hfill \atop
  \scriptstyle k_i  \in \{ 1, \ldots ,n_{l,i} \} ,i \ne c \hfill} d^2_{B^H } (J_l ,K) ,\sum\limits_{\scriptstyle k_c  \in \{ 1, \ldots ,u_h \} , \hfill \atop
  \scriptstyle k_i  \in \{ 1, \ldots ,n_{h,i} \} ,i \ne c \hfill} d^2_{B^H } (J_h ,K) } \right)} ], \label{Eq:covSqlh}
\end{align}
$1\leq l,h \leq m $, and $u_l  =  2^{j_{m,c}-j_{l,c}}$. $n_{J_l}=  \prod\nolimits_{i=1}^d n_{l,i}$, $l \in \{1,\ldots,m \}$.
\end{lemma}

$\sigma _{l,h} ^2$ can be calculated by $\sigma _{aS } ^2$ in Lemma \ref{Lm:aSq}.

Note that $\sigma _{l,l} ^2$ is equal to $\sigma^2 (J_l)/u_l$ in Lemma \ref{Lm:SqJ}. In fact,
\begin{align*}
 \sigma _{l,l} ^2
  &= \frac{1}{{(u_l n_{J_l } /n_{l,c} )^2 }}[{\mathop{\rm Var}} \left( {\sum\limits_{\scriptstyle k_c  \in \{ 1, \ldots ,u_l \} , \hfill \atop
  \scriptstyle k_i  \in \{ 1, \ldots ,n_{l,i} \} ,i \ne c \hfill} {d^2_{B^H } (J_l ,K)} } \right) \\
  &+ 2\sum\limits_{k = 1}^\infty  {{\mathop{\rm Cov}} \left( {\sum\limits_{\scriptstyle k_c  \in \{ 1, \ldots ,u_l \} , \hfill \atop
  \scriptstyle k_i  \in \{ 1, \ldots ,n_{l,i} \} ,i \ne c \hfill} {d^2_{B^H } (J_l ,K) } ,\sum\limits_{\scriptstyle k_c  \in \{ u_l k + 1, \ldots ,u_l (k + 1)\} , \hfill \atop
  \scriptstyle k_i  \in \{ 1, \ldots ,n_{l,i} \} ,i \ne c \hfill} {d^2_{B^H } (J_l ,K) } } \right)} ] \\
  &= \frac{1}{{(u_l n_{J_l } /n_{l,c} )^2 }}[\sum\limits_{j = 1}^{u_l } {{\mathop{\rm Var}} \left( {\sum\limits_{k_c  = j,k_i  \in \{ 1, \ldots ,n_{l,i} \} ,i \ne c} {d^2_{B^H } (J_l ,K) } } \right)}  \\
  &+ \sum\limits_{j = 1}^{u_l } {2\sum\limits_{k = 1}^\infty  {{\mathop{\rm Cov}} \left( {\sum\limits_{k_c  = j,k_i  \in \{ 1, \ldots ,n_{l,i} \} ,i \ne c} {d^2_{B^H } (J_l ,K) } ,\sum\limits_{k_c  = j + k,k_i  \in \{ 1, \ldots ,n_{l,i} \} ,i \ne c} {d^2_{B^H } (J_l ,K) } } \right)} } ].
\end{align*}
Since the sequence of random vectors $X_v$ in Lemma \ref{Lm:SqJ} is stationary,
\begin{align*}
 \sigma _{l,l} ^2
 & = \frac{1}{{u_l (n_{J_l } /n_{l,c} )^2 }}[{\mathop{\rm Var}} \left( {\sum\limits_{k_c  = 1,k_i  \in \{ 1, \ldots ,n_{l,i} \} ,i \ne c} {d^2_{B^H } (J_l ,K)} } \right) \\
 &+ 2\sum\limits_{k = 1}^\infty  {{\mathop{\rm Cov}} \left( {\sum\limits_{k_c  = 1,k_i  \in \{ 1, \ldots ,n_{l,i} \} ,i \ne c} {d^2_{B^H } (J_l ,K) } ,\sum\limits_{k_c  = 1 + k,k_i  \in \{ 1, \ldots ,n_{l,i} \} ,i \ne c} {d^2_{B^H } (J_l ,K) } } \right)} ] \\
 &= \sigma ^2 (J_l) / u_l .
\end{align*}

By Theorem \ref{Thm:funclimit}, one gets the asymptotic behavior of the logarithm of $S$.

\begin{lemma}\label{Lm:LSq}
Under the conditions of Lemma \ref{Lm:Tnj}, $L = (\log_2 S (J_1),\ldots, \log_2 S(J_m))'$. $N>\mathop{\max }\limits_{1 \le i \le d} H_i+1/4$.
For any $c \in \{1, \ldots, d \}$, $T_c \to +\infty $, we have
\begin{align}
n_{m,c } ^{1/2} [L  -  Lu  ]\mathop \to\limits^d  N(0,\Sigma _{L }),
\end{align}
where
\begin{align}
Lu  = (\log_2 \mathbb{E}d^2_{B^H } (J_1 ,K ) , \ldots , \log_2 \mathbb{E}d^2_{B^H } ( J_m ,K ) ),
\end{align}
and
\begin{align}
\Sigma _{L }  = D' \Sigma _{S_q } D.
\end{align}
with $D = \log _2 e \cdot {\rm diag}(\left( { \mathbb{E} d^2_{B^H } (J_1 ,K) } \right)^{ - 1} , \ldots ,\left( { \mathbb{E} d^2_{B^H } (J_m ,K) } \right)^{ - 1} )$, and $\Sigma _{S }$ is defined by (\ref{Eq:sigmaS1}) and (\ref{Eq:covSqlh}).

the same as that in Lemma \ref{Lm:Sq}.
\end{lemma}

The estimator $\hat H$ denoted by Eq.(\ref{Eq:Hesti}) is linear combinations of the elements of $L$.
So the asymptotic normality of $\hat H$ follows from the asymptotic normality of $L$ by Theorem \ref{Thm:funclimit}.


\subsection{Two-Step estimator}\label{Sec:twostepesti}

We adopt the GLS estimator, which has the minimum variance in these estimators defined by (\ref{Eq:Hesti}) (see Gauss-Markov theorem \cite{WLS}).
But the covariance matrix ${\mathop{\rm Var}} L$ used by GLS estimator is a function of the unknown parameter $H$ we want to estimate.
So a two-step estimator may be feasible, namely, we first obtain the OLS estimator of $H$, and then use the OLS estimator to estimate the covariance matrix. So the GLS estimator can be obtained using the estimate of the covariance matrix.
This estimator has been used to the estimation of Hurst parameter of fBm and is proved to be asymptotically normal (see \cite{Bardet2000,Bardet2002,Morales2002thesis}).

In this subsection, we apply two-step estimator to the case of fractional Brownian sheet.

First we obtain the covariance matrix of $L$ with known $H$. By the regularity on $R_+$ of logarithm function and Theorem \ref{Thm:funclimit}, when $n_{J_h}, n_{J_l}$ is large, the $(h,l)$ element of the covariance matrix ${\mathop{\rm Var}} L $ is given by
\begin{align}
 {\rm{Cov}}(\log _2 S (J_h ),\log _2 S (J_l )) &\approx \frac{{{\rm{(log}}_2 {\rm{e)}}^2 {\rm{Cov}}(S (J_h ),S (J_l ))}}{{\mathbb{E}d^2 _{B^H } (J_h ,K)\mathbb{E}d^2 _{B^H } (J_l ,K)}} \notag \\
  &= {\rm{(log}}_2 {\rm{e)}}^2 \frac{{\sum\limits_{K_l } {\sum\limits_{K_h } {{\rm{Cov}}(d^2 _{B^H } (J_h ,K_h ),d^2 _{B^H } (J_l ,K_l ))} } }}{{n_{J_h } n_{J_l } \mathbb{E}d^2 _{B^H } (J_h ,K)\mathbb{E}d^2 _{B^H } (J_l ,K)}} \notag \\
  &= {\rm{(log}}_2 {\rm{e)}}^2 \frac{{\sum\limits_{K_l } {\sum\limits_{K_h } {{\rm{2Cov}}^2 (d_{B^H } (J_h ,K_h ),d_{B^H } (J_l ,K_l ))} } }}{{n_{J_h } n_{J_l } \mathbb{E}d^2 _{B^H } (J_h ,K)\mathbb{E}d^2 _{B^H } (J_l ,K)}}. \text{(by Lemma \ref{Lm:covXY2})} \label{Eq:GLScov}
\end{align}
Consider ${\rm{Cov}}(d_{B^H } (J_h ,K_h ),d_{B^H } (J_l ,K_l ))$, by Eq.(\ref{Eq:P1cov}),
\begin{align}
{\rm{Cov}}(d_{B^H } (J_h ,K_h ),d_{B^H } (J_l ,K_l )) = \prod\limits_{i = 1}^d {\int_R {\int_R {\psi _{j_{h.i} ,k_{h,i} } (t)\psi _{j_{l.i} ,k_{l,i} } (s)( - |t - s|^{2H_i} /2)dtds} } },  \label{Eq:GLSdcov}
\end{align}
where $J_l=(j_{l,1},\ldots,j_{l,d}) \in \mathbb{Z}^d_+$, $J_h=(j_{h,1},\ldots,j_{h,d}) \in \mathbb{Z}^d_+$, $K_l=(k_{l,1},\ldots,k_{l,d}) \in \mathbb{Z}^d_+$ and $K_h=(k_{h,1},\ldots,k_{h,d}) \in \mathbb{Z}^d_+$, $l,h\in \{1,\ldots,m \}$.

Using Eq.(\ref{Eq:GLSdcov}), ${\rm{Cov}}(d_{B^H } (J_h ,K_h ),d_{B^H } (J_l ,K_l ))$ can be calculated via two-dimensional wavelet transform of function $ g(t,s)= - |t - s|^{2H_i} /2$ for each $i \in \{1,\ldots,d \}$.
$\mathbb{E}d^2 _{B^H } (J_h ,K)$ and $\mathbb{E}d^2 _{B^H } (J_l ,K)$ can be calculated similarly.

When $H$ is known, the covariance matrix of $L$ can be calculated approximately.
We denote the approximately calculated covariance matrix by $G(H)$, which is a function of $H$.

In view of Eq.(\ref{Eq:GLScov}) and Eq.(\ref{Eq:GLSdcov}), $G(H)$ is a positive definite symmetric matrix and differentiable in $H$ for $H \in (0,1)^d$. 
So $G(H)^{-1}$  exists and is continuous in $H$ for $H \in (0,1)^d$.
Now, we can write the two-step estimator,
\begin{align}\label{Eq:TSesti}
 & \hat \alpha = (A' G(\hat H_o)^{-1} A)^{ - 1} A G(\hat H_o)^{-1} L, \notag \\
 & \hat H_{og}  = \hat \alpha/2  - 1/2,
\end{align}
where $\hat H_o$ is the OLS estimator.

\begin{theorem}\label{Thm:TSesti}
The two-step estimator $\hat H_{og}$ denoted by Eq.(\ref{Eq:TSesti}) is asymptotically normal, and has the same asymptotic covariance as the GLS estimator.
\end{theorem}

\begin{proof}
The asymptotic normality follow from Theorem \ref{Thm:Hmain}.

By Eq.(\ref{Eq:GLScov}) and Eq.(\ref{Eq:GLSdcov}), it's easy to check that $G(H)$ is a positive symmetric matrix and differentiable in $H$ for $H \in (0,1)$. And from Theorem \ref{Thm:Hmain}, for any $c \in \{1, \ldots, d \}$, $\hat H_o \mathop {\mathop  \to \limits_{T_c  \to \infty } }\limits^P H$. So we have
\begin{align}\label{Eq:TSThm1}
G(\hat H_o )^{ - 1} \mathop {\mathop  \to \limits_{T_c  \to \infty } }\limits^P G(H)^{ - 1}.
\end{align}

Another, $G(H)= D' {\rm Var(S)} D$, where $D = \log _2 e \cdot {\rm diag}(\left( { \mathbb{E} d^2_{B^H } (J_1 ,K) } \right)^{ - 1} , \ldots ,\left( { \mathbb{E} d^2_{B^H } (J_m ,K) } \right)^{ - 1} )$, $S = (S (J_1),\ldots,S(J_m))'$.
By Lemma \ref{Lm:Sq}, ${\rm Var(S)} \mathop  \to \limits_{T_c  \to \infty } \Sigma _{S }/n_{m,c}$. So
\begin{align}\label{Eq:TSThm2}
G(H) \mathop  \to \limits_{T_c  \to \infty } \Sigma _{L }/n_{m,c}.
\end{align}
$\Sigma _{L } / n_{m,c} $ is the asymptotic covariance of $L$.

Combining (\ref{Eq:TSThm1}) and (\ref{Eq:TSThm2}), the two-step estimator has the same asymptotic covariance as the GLS estimator.
\end{proof}

\section{Simulation results and discussion} \label{Sec:simulation}

In this section, we use the estimator $\hat H_{og}$ to estimate the Hurst parameter of generated fBs.
The data we test are generated by the method of circulant embedding of the covariance matrix \cite{FBSsim9,FBSsim4}.

\paragraph{Selection of parameters}
Before the estimation, the octaves $J$ and the number of vanishing moments $N$ (or wavelet) must be chosen.
For octaves $J=(j_1,\cdots,j_d)$, we choose all the octaves $J$ that satisfy $J_1 \leq J \leq J_m$, $J_1=(j_{1,1} , \cdots ,j_{1,d} )$ is the lower bound of $J$, $J_m=(j_{m,1} , \cdots ,j_{m,d} )$ is the lower bound of $J$.
Based on prior studies \cite{FBMwave7,FBMwave5}, for each $i\in \{1, \ldots, d \}$, $j_{m,i}$ is chosen the largest possible given the data length of dimension $i$.
$j_{1,i}$ is chosen $3$ by the minimum mean square error (MSE).
It is known that there is a bias-variance trade-off for choosing octaves: the selection of small octave increase the bias, but decrease the variance of the estimator.
The MSE allows the tradeoff between variance and bias.

We choose the classical Daubechies wavelet, which are orthonormal and have compact support.
From Theorem \ref{Thm:Hmain}, the number of vanishing moments must be chosen $N \geq 2$.
An increase in the number of vanishing moments $N$ comes with an enlargement of the compact support ($[-M,M]$)\cite{FBMwave7}.
In the case of finite time, this will lead to the decrease of the number of available wavelet coefficients $n_J$ at each octaves.
So we choose $N=3$.

\paragraph{Estimation performance}
The estimator $\hat H_{og}$ for each $H$ is applied to 500 independent copies of size $512\times 512$ for the two-dimensional case and 100 independent copies of size $256\times 256\times 256$ for three-dimensional case. All the results are shown in Tab.\ref{Tab:estimation}.

Tab.\ref{Tab:estimation} shows that both $\hat H_o$ and $\hat H_{og}$ have small MSE.
For all the estimated $H$, the means of $\hat H_o$ and $\hat H_{og}$ are both less than $H$, and the two estimators have the similar bias, but the variance and MSE of $\hat H_{og}$ are less than that of $\hat H_o$.
Besides, we can also find that the smaller the $H_i$ is, the larger the bias and MSE are.

\paragraph{Discussion}
From our simulation results, we can conclude that both two-step estimator $\hat H_{og}$ and OLS estimator $\hat H_o$ perform very well in the sense of small MSE.

It can be also seen that  $\hat H_{og}$ perform better than $\hat H_o$ because $\hat H_{og}$ has smaller variance than $\hat H_o$,
which is consistent with  the theoretical result we have discussed in Section \ref{Sec:twostepesti}.

Both estimators are not unbiased because
\begin{eqnarray*}
\mathbb{E}\log _2 S(J) \ne \log _2 \mathbb{E}S(J) = \log _2 \mathbb{E}d_{B^H } ^2 (J,K).
\end{eqnarray*}
and all the theorems are based on the case of continuous time, but the data we used is discrete,
the means of both estimators are less than truth parameters as we can see the simulation results.
The bias becomes larger as the value of $H$ becomes smaller.
For $H_i \geq 0.5$, the bias of $\hat H_i$ is small and can be ignored.
For $H_i < 0.5$, the bias of $\hat H_i$ is obvious but still small.
The estimator for $H_i \geq 0.5$ performs better than that for $H_i < 0.5$.

\begin{table}[htbp]
\centering
\caption{Estimation performance for $\hat H_o$ and $\hat H_{og}$}
\begin{tabular}{cc|ccc|ccc}
\hline
           &            &    \multicolumn{ 3}{|c}{$\hat H_o$} & \multicolumn{ 3}{|c}{$\hat H_{og}$} \\
 Dimension &       $H$ &       Mean &        Std &       RMSE &       Mean &        Std &       RMSE \\
\hline
\multicolumn{ 1}{c}{2D} & $H_1=0.3$ &    0.2733  &    0.0297  &    0.0399  &    0.2706  &    0.0144  &    0.0327  \\
\multicolumn{ 1}{c}{} & $H_2=0.3$ &    0.2732  &    0.0291  &    0.0396  &    0.2705  &    0.0139  &    0.0326  \\
\hline
\multicolumn{ 1}{c}{2D} & $H_1=0.5$ &    0.4897  &    0.0301  &    0.0318  &    0.4911  &    0.0150  &    0.0174  \\
\multicolumn{ 1}{c}{} & $H_2=0.5$ &    0.4855  &    0.0322  &    0.0353  &    0.4889  &    0.0155  &    0.0191  \\
\hline
\multicolumn{ 1}{c}{2D} & $H_1=0.8$ &    0.7915  &    0.0323  &    0.0334  &    0.7958  &    0.0165  &    0.0170  \\
\multicolumn{ 1}{c}{} & $H_2=0.8$ &    0.7917  &    0.0344  &    0.0354  &    0.7961  &    0.0171  &    0.0176  \\
\hline
\multicolumn{ 1}{c}{2D} & $H_1=0.3$ &    0.2731  &    0.0303  &    0.0405  &    0.2708  &    0.0140  &    0.0324  \\
\multicolumn{ 1}{c}{} & $H_2=0.5$ &    0.4865  &    0.0312  &    0.0340  &    0.4892  &    0.0146  &    0.0181  \\
\hline
\multicolumn{ 1}{c}{2D} & $H_1=0.3$ &    0.2736  &    0.0324  &    0.0418  &    0.2719  &    0.0153  &    0.0320  \\
\multicolumn{ 1}{c}{} & $H_2=0.8$ &    0.7931  &    0.0296  &    0.0304  &    0.7980  &    0.0150  &    0.0151  \\
\hline
\multicolumn{ 1}{c}{2D} & $H_1=0.5$ &    0.4878  &    0.0339  &    0.0360  &    0.4901  &    0.0164  &    0.0192  \\
\multicolumn{ 1}{c}{} & $H_2=0.8$ &    0.7911  &    0.0337  &    0.0348  &    0.7971  &    0.0152  &    0.0155  \\
\hline
\multicolumn{ 1}{c}{2D} & $H_1=0.7$ &    0.6886  &    0.0349  &    0.0367  &    0.6937  &    0.0174  &    0.0185  \\
\multicolumn{ 1}{c}{} & $H_2=0.8$ &    0.7902  &    0.0334  &    0.0347  &    0.7959  &    0.0169  &    0.0174  \\
\hline
\multicolumn{ 1}{c}{} & $H_1=0.6$ &    0.5929  &    0.0171  &    0.0184  &    0.5943  &    0.0085  &    0.0102  \\
\multicolumn{ 1}{c}{3D} & $H_2=0.7$ &    0.6942  &    0.0172  &    0.0181  &    0.6973  &    0.0080  &    0.0084  \\
\multicolumn{ 1}{c}{} & $H_3=0.8$ &    0.7955  &    0.0152  &    0.0158  &    0.7970  &    0.0087  &    0.0091  \\
\hline
\multicolumn{ 1}{c}{} & $H_1=0.7$ &    0.6941  &    0.0186  &    0.0195  &    0.6967  &    0.0095  &    0.0100  \\
\multicolumn{ 1}{c}{3D} & $H_2=0.8$ &    0.7931  &    0.0181  &    0.0193  &    0.7959  &    0.0089  &    0.0098  \\
\multicolumn{ 1}{c}{} & $H_3=0.9$ &    0.8962  &    0.0167  &    0.0171  &    0.8980  &    0.0105  &    0.0106  \\
\hline
\end{tabular}

\footnotesize The estimation performance for OLS estimator is given on the left, and for two-step estimator is given on the right.
Std denotes standard deviation, RMSE denotes root of MSE.
These values are obtained from 500 realizations of generated fBs of size $512\times 512$, or 100 realizations of generated fBs of size $256\times 256\times 256$.
It is shown that the means of $\hat H_o$ and $\hat H_{og}$ are both less than $H$, and the two estimators have the similar bias, but the Std and RMSE of $\hat H_{og}$ are less than that of $\hat H_o$ (for more detail see Section \ref{Sec:simulation}).
\label{Tab:estimation}
\end{table}

\section{Conclusions} \label{Sec:conclusion}
In this paper, we proposed a class of statistical estimators $\hat H =(\hat H_1, \ldots, \hat H_d)$ for the Hurst parameters $H=(H_1, \ldots, H_d)$ of fractional Brownian field via multidimensional wavelet analysis and least squares.
The proposed estimators are based on the wavelet expansion of fractional Brownian sheet and a regression on the log-variance of wavelet coefficient.
We also prove that our estimators $\hat H$ are asymptotically normal following the approach in \cite{Morales2002thesis}.
The convergence rate of $\hat H$ indicates that the estimation accuracy of $\hat H$ can be improved by increasing the square root of sample volume $\sqrt{\prod\nolimits_{i=1}^d {T_i}}$.
The main difficulty of our proof is the asymmetric normality of linear combination of $S$ (see Lemma \ref{Lm:aSq}), which is achieved by the construction of the sequence of random vectors $\{ X_v \}_{v=1}^{\infty}$ from wavelet coefficients of fBs and the function $h(X_v)$ that satisfy the condition of Theorem \ref{Thm:Arcones}.
Using the two-step procedure, we realize the generalized least squares (GLS) estimator, which has the lowest variance of $\hat H$.
We simulated some fractional Brownian sheets and used them to validate the two-step estimator.
We find that when $H_i \geq 1/2$, the estimators are efficient, and when $H_i < 1/2$, there are some bias.

The computation code of our method is available upon request.
%

\ifCLASSOPTIONcaptionsoff
  \newpage
\fi



\bibliographystyle{IEEEtran}
\bibliography{wavebibfileTIT}

\begin{thebibliography}{10}
\providecommand{\url}[1]{#1}
\csname url@samestyle\endcsname
\providecommand{\newblock}{\relax}
\providecommand{\bibinfo}[2]{#2}
\providecommand{\BIBentrySTDinterwordspacing}{\spaceskip=0pt\relax}
\providecommand{\BIBentryALTinterwordstretchfactor}{4}
\providecommand{\BIBentryALTinterwordspacing}{\spaceskip=\fontdimen2\font plus
\BIBentryALTinterwordstretchfactor\fontdimen3\font minus
  \fontdimen4\font\relax}
\providecommand{\BIBforeignlanguage}[2]{{%
\expandafter\ifx\csname l@#1\endcsname\relax
\typeout{** WARNING: IEEEtran.bst: No hyphenation pattern has been}%
\typeout{** loaded for the language `#1'. Using the pattern for}%
\typeout{** the default language instead.}%
\else
\language=\csname l@#1\endcsname
\fi
#2}}
\providecommand{\BIBdecl}{\relax}
\BIBdecl

\bibitem{FBMwave2}
P.~Abry, P.~Gon\c{c}alv\`{e}s, and P.~Flandrin, \emph{Wavelets, spectrum
  analysis and 1/f processes}, ser. Lecture Notes in Statistics.\hskip 1em plus
  0.5em minus 0.4em\relax Springer New York, 1995, vol. 103, book section~2,
  pp. 15--29.

\bibitem{FBMwave1}
P.~Flandrin, ``Wavelet analysis and synthesis of fractional brownian motion,''
  \emph{Information Theory, IEEE Transactions on}, vol.~38, no.~2, pp.
  910--917, 1992.

\bibitem{FBMwave3}
L.~Delbeke and W.~Van~Assche, ``A wavelet based estimator for the parameter of
  self-similarity of fractional brownian motion,'' in \emph{3rd International
  Conference on Approximation and Optimization in the Caribbean (Puebla,
  1995)}, vol.~24, Conference Proceedings, pp. 65--76.

\bibitem{FBMwave4}
P.~Abry and D.~Veitch, ``Wavelet analysis of long-range-dependent traffic,''
  \emph{Information Theory, IEEE Transactions on}, vol.~44, no.~1, pp. 2--15,
  1998.

\bibitem{FBMwave5}
D.~Veitch and P.~Abry, ``A wavelet-based joint estimator of the parameters of
  long-range dependence,'' \emph{Information Theory, IEEE Transactions on},
  vol.~45, no.~3, pp. 878--897, 1999.

\bibitem{FBMwave6}
P.~Abry, P.~Flandrin, M.~S. Taqqu, and D.~Veitch, ``Wavelets for the analysis,
  estimation and synthesis of scaling data,'' \emph{Self-similar network
  traffic and performance evaluation}, pp. 39--88, 2000.

\bibitem{FBMwave7}
------, ``Self-similarity and long-range dependence through the wavelet lens,''
  \emph{Theory and applications of long-range dependence}, pp. 527--556, 2003.

\bibitem{FBMwave8}
\BIBentryALTinterwordspacing
P.~Abry, H.~Helgason, and V.~Pipiras, ``Wavelet-based analysis of non-gaussian
  long-range dependent processes and estimation of the hurst parameter,''
  \emph{Lithuanian Mathematical Journal}, vol.~51, no.~3, pp. 287--302, 2011.
  [Online]. Available: \url{http://dx.doi.org/10.1007/s10986-011-9126-4}
\BIBentrySTDinterwordspacing

\bibitem{Ding2014}
Y.~Ding, Y.~Hu, W.~Xiao, and L.~Yan, ``Long-memory processes and
  applications,'' \emph{Abstract and Applied Analysis}, vol. 2014, 2014.

\bibitem{FBSdef1}
A.~Kamont, ``On the fractional anisotropic wiener field,'' \emph{Probability
  and Mathematical Statistics-PWN}, vol.~16, no.~1, pp. 85--98, 1996.

\bibitem{Pesquet2002}
B.~Pesquet-Popescu and J.~L. Vehel, ``Stochastic fractal models for image
  processing,'' \emph{Signal Processing Magazine, IEEE}, vol.~19, no.~5, pp.
  48--62, 2002.

\bibitem{Bierme2007}
\BIBentryALTinterwordspacing
H.~Bierm\'{e}, M.~M. Meerschaert, and H.-P. Scheffler, ``Operator scaling
  stable random fields,'' \emph{Stochastic Processes and their Applications},
  vol. 117, no.~3, pp. 312--332, 2007. [Online]. Available:
  \url{http://www.sciencedirect.com/science/article/pii/S0304414906001104}
\BIBentrySTDinterwordspacing

\bibitem{Bonami2003}
\BIBentryALTinterwordspacing
A.~Bonami and A.~Estrade, ``Anisotropic analysis of some gaussian models,''
  \emph{Journal of Fourier Analysis and Applications}, vol.~9, no.~3, pp.
  215--236, 2003. [Online]. Available:
  \url{http://dx.doi.org/10.1007/s00041-003-0012-2}
\BIBentrySTDinterwordspacing

\bibitem{Doukhan2003book}
P.~Doukhan, G.~Oppenheim, and M.~S. Taqqu, \emph{Theory and applications of
  long-range dependence}.\hskip 1em plus 0.5em minus 0.4em\relax Springer,
  2003.

\bibitem{FBSstu3}
\BIBentryALTinterwordspacing
D.~Wu and Y.~Xiao, ``Geometric properties of fractional brownian sheets,''
  \emph{Journal of Fourier Analysis and Applications}, vol.~13, no.~1, pp.
  1--37, 2007. [Online]. Available:
  \url{http://dx.doi.org/10.1007/s00041-005-5078-y}
\BIBentrySTDinterwordspacing

\bibitem{Atto2013}
A.~M. Atto, Y.~Berthoumieu, and P.~Bolon, ``2-d wavelet packet spectrum for
  texture analysis,'' \emph{Image Processing, IEEE Transactions on}, vol.~22,
  no.~6, pp. 2495--2500, 2013.

\bibitem{Roux2013}
S.~G. Roux, M.~Clausel, B.~Vedel, S.~Jaffard, and P.~Abry, ``Self-similar
  anisotropic texture analysis: The hyperbolic wavelet transform
  contribution,'' \emph{Image Processing, IEEE Transactions on}, vol.~22,
  no.~11, pp. 4353--4363, 2013.

\bibitem{book2}
F.~Biagini, Y.~Hu, B.~{\O}ksendal, and T.~Zhang, \emph{Stochastic calculus for
  fractional Brownian motion and applications}.\hskip 1em plus 0.5em minus
  0.4em\relax Springer London, 2008.

\bibitem{Moseley2002}
\BIBentryALTinterwordspacing
M.~Moseley, ``Diffusion tensor imaging and aging – a review,'' \emph{NMR in
  Biomedicine}, vol.~15, no. 7-8, pp. 553--560, 2002. [Online]. Available:
  \url{http://dx.doi.org/10.1002/nbm.785}
\BIBentrySTDinterwordspacing

\bibitem{Bardet2002}
J.~M. Bardet, ``Statistical study of the wavelet analysis of fractional
  brownian motion,'' \emph{Information Theory, IEEE Transactions on}, vol.~48,
  no.~4, pp. 991--999, 2002.

\bibitem{Morales2002thesis}
C.~J. Morales, ``Wavelet-based multifractal spectra estimation: Statistical
  aspects and applications,'' Ph.D. thesis, Boston University, 2002.

\bibitem{Bardet2000}
\BIBentryALTinterwordspacing
J.-M. Bardet, ``Testing for the presence of self-similarity of gaussian time
  series having stationary increments,'' \emph{Journal of Time Series
  Analysis}, vol.~21, no.~5, pp. 497--515, 2000. [Online]. Available:
  \url{http://dx.doi.org/10.1111/1467-9892.00195}
\BIBentrySTDinterwordspacing

\bibitem{FBMdef}
\BIBentryALTinterwordspacing
B.~Mandelbrot and J.~Van~Ness, ``Fractional brownian motions, fractional noises
  and applications,'' \emph{SIAM Review}, vol.~10, no.~4, pp. 422--437, 1968.
  [Online]. Available: \url{http://epubs.siam.org/doi/abs/10.1137/1010093}
\BIBentrySTDinterwordspacing

\bibitem{FBSdef2}
Y.~Hu, ``Heat equations with fractional white noise potentials,'' \emph{Applied
  Mathematics and Optimization}, vol.~43, no.~3, pp. 221--243, 2001.

\bibitem{FBSsim2}
A.~Brouste, J.~Istas, and S.~Lambert-Lacroix, ``On fractional gaussian random
  fields simulations,'' \emph{Journal of Statistical Software}, vol.~23, no.~1,
  pp. 1--23, 2007.

\bibitem{Basu2004book}
A.~K. Basu, \emph{Measure theory and probability}.\hskip 1em plus 0.5em minus
  0.4em\relax PHI Learning Pvt. Ltd., 2004.

\bibitem{Arcones1994}
M.~A. Arcones, ``Limit theorems for nonlinear functionals of a stationary
  gaussian sequence of vectors,'' \emph{The Annals of Probability}, vol.~22,
  no.~4, pp. 2242--2274, 1994.

\bibitem{Serfling1980book}
R.~J. Serfling, \emph{Approximation theorems of mathematical statistics}.\hskip
  1em plus 0.5em minus 0.4em\relax John Wiley \& Sons, 1980, vol. 162.

\bibitem{book3}
I.~Daubechies, \emph{Ten lectures on wavelets}.\hskip 1em plus 0.5em minus
  0.4em\relax SIAM, 1992, vol.~61.

\bibitem{book4}
J.~C. Goswami and A.~K. Chan, \emph{Fundamentals of wavelets: theory,
  algorithms, and applications}.\hskip 1em plus 0.5em minus 0.4em\relax John
  Wiley \& Sons, 2011, vol. 233.

\bibitem{FBMwave9}
J.-M. Bardet, G.~Lang, G.~Oppenheim, A.~Philippe, S.~Stoev, and M.~S. Taqqu,
  ``Semi-parametric estimation of the long-range dependence parameter: a
  survey,'' \emph{Theory and applications of long-range dependence}, pp.
  557--577, 2003.

\bibitem{FBMwave11a}
A.~H. Tewfik and M.~Kim, ``Correlation structure of the discrete wavelet
  coefficients of fractional brownian motion,'' \emph{IEEE transactions on
  information theory}, vol.~38, no.~2, pp. 904--909, 1992.

\bibitem{FBMwave11b}
R.~W. Dijkerman and R.~R. Mazumdar, ``On the correlation structure of the
  wavelet coefficients of fractional brownian motion,'' \emph{Information
  Theory, IEEE Transactions on}, vol.~40, no.~5, pp. 1609--1612, 1994.

\bibitem{WLS}
A.~C. Aitken, ``On least squares and linear combinations of observations,''
  \emph{Proceedings of the Royal Society of Edinburgh}, vol.~55, pp. 42--48,
  1935.

\bibitem{FBSsim9}
R.~B. Davies and D.~Harte, ``Tests for hurst effect,'' \emph{Biometrika},
  vol.~74, no.~1, pp. 95--101, 1987.

\bibitem{FBSsim4}
D.~P. Kroese and Z.~I. Botev, ``Spatial process generation,'' \emph{arXiv
  preprint arXiv:1308.0399}, 2013.

\end{thebibliography}
\end{document}